\documentclass[reqno]{amsart}
\usepackage{amsmath,amssymb,amsthm,amsfonts}
\usepackage{color}
\usepackage{mathrsfs}
\usepackage{hyperref}
\usepackage{enumerate}

\theoremstyle{plain}
\newtheorem{theorem}{Theorem}
\newtheorem{lemma}[theorem]{Lemma}
\newtheorem{prop}[theorem]{Proposition}

\theoremstyle{definition}
\newtheorem{definition}[theorem]{Definition}

\theoremstyle{remark}
\newtheorem*{remark}{Remark}

\newcommand{\e}{\mathrm{e}} 
\newcommand{\N}{\mathbb{N}}
\newcommand{\R}{\mathbb{R}}

\newcommand{\dd}{\mathrm{d}} 

\newcommand{\vect}[1]{\boldsymbol{#1}}

\newcommand{\be}{\begin{equation}}
\newcommand{\ee}{\end{equation}}
\newcommand{\ba}{\begin{equation} \begin{aligned}}
\newcommand{\ea}{\end{aligned}\end{equation}}
\newcommand{\bes}{\begin{equation*}}
\newcommand{\ees}{\end{equation*}}

\def\1{{\mathchoice {1\mskip-4mu\mathrm l}      
{1\mskip-4mu\mathrm l}
{1\mskip-4.5mu\mathrm l} {1\mskip-5mu\mathrm l}}}

\title{Revisiting Groeneveld's approach to the virial expansion}
\author{Sabine Jansen} 
\address{Mathematisches Institut, Ludwig-Maximilians-Universit{\"a}t, Theresienstr. 39, 80333 M{\"u}nchen; Munich Center for Quantum Science and Technology (MCQST), Schellingstr. 4, 80799 M{\"u}nchen, Germany.}
\email{jansen@math.lmu.de}

\date{19 September 2020}

\begin{document}


\begin{abstract}
	A generalized version of Groeneveld's convergence criterion for the virial expansion and generating functionals for weighted $2$-connected graphs is proven. The criterion works for inhomogeneous systems and yields bounds for the density expansions of the correlation functions $\rho_s$ (a.k.a.\ distribution functions or factorial moment measures) of grand-canonical Gibbs measures with pairwise interactions. The proof is based on recurrence relations for graph weights related to the Kirkwood-Salsburg integral equation for correlation functions. The proof does not use an inversion of the density-activity expansion, however a M{\"o}bius inversion on the lattice of set partitions enters the derivation of the recurrence relations. 	\\

\noindent \emph{Keywords}: equilibrium statistical mechanics; cluster expansions;  Kirkwood-Salsburg integral equation; 2-connected graphs and their generating functions.\\

\noindent \emph{Mathematics Subject Classification (2020)}: 82B05; 05A15.
\end{abstract}

\maketitle

\tableofcontents

\section{Introduction}

Graphical expansions of thermodynamic functionals and correlation functions play an important role in statistical mechanics and liquid state theory. Relevant quantities are expanded in powers of the activity $z$ or the density $\rho$, leading for example to the Mayer expansion and virial expansion for the pressure. In mathematical statistical physics, the expansions are used to establish absence of phase transitions and exponential decay of correlations~\cite{ruelle1969book}, though in this regard  disagreement percolation \cite{georgi-haggstrom-maes2001,dereudre2019survey,hofer-houdebert2019,benes-hofer-last-vecera2020}, Dobrushin uniqueness \cite{houdebert-zass2020dobrushin}, approaches based on Glauber birth and death dynamics or other algorithms \cite{ferrari-fernandez-garcia2002,fernandez-groisman-saglietti2016,helmuth-perkins-petti2020}, or recursive approaches and complex analysis \cite{meeron1970,michelen-perkins2020} often yield better results; in addition to expansions, there are also bounds \cite{lieb1963newmethod}.  In physical chemistry and density functional theory, diagrammatic expansions serve as a conceptual guide to various approximation schemes and closure relations notably for the Ornstein-Zernike equation, see \cite{hansen-mcdonald2013} and the discussion and references in~\cite{kuna-tsagkaro2018}.

It has been conjectured that the virial expansion converges in a bigger domain than the activity expansion~\cite{brydges2011iamp,groeneveld1967c}. Examples where this is proven include hard hexagons~\cite{joyce1988}, hard rods on a line (Tonks gas)~\cite{tonks1936,jansen2015tonks}, uniformly repulsive interactions in finite volume~\cite{brydges-marchetti2014, marchetti2015}, and a hierarchical mixture of cubes on a lattice \cite{jansen2020hierarchical}. A counter-example for an attractive double-well potential that favors dimerization of particles \cite{jansen2012mayer} suggests that the conjecture, if true at all, perhaps only applies to non-negative potentials. 

The best-known proofs of convergence for the virial expansion are based on an inversion. First one proves convergence of the activity expansion in powers of $z$, then one inverts the density-activity expansion $\rho = \rho(z)$ to obtain an expansion in powers of $\rho$~\cite{lebowitz-penrose1964}. Clearly this approach is ill-suited to a proof or disproof of the above-mentioned conjecture: any  convergence criterion for density expansions derived in this way inherits the limitations of activity expansions. 

Alternative approaches were given by Groeneveld~\cite{groeneveld1967c}, Pulvirenti and Tsagka\-rogiannis~\cite{pulvirenti-tsagkaro2012}, Ramawadh and Tate~\cite{ramawadh-tate2015}, and Nguyen and Fern{\'a}ndez~\cite{nguyen-fernandez2020}. Groeneveld's proof is based on recurrence relations for weighted graphs upon successive removal of edges incident to a given vertex.  Pulvirenti and Tsagkarogiannis tackled convergence directly in the canonical ensemble. Ramawadh, Tate, Nguyen, and Fern{\'a}ndez exploited Lagrange inversion and combinatorial features of the Penrose tree partition scheme. 

The aim of the present note is to revisit Groeneveld's approach, which may help address the above-mentioned conjecture in the future. Our approach is based on a recurrence relation similar to Groeneveld's except we remove a vertex and all its incident edges in one go (Proposition~\ref{prop:recursion}). The recurrence relations obtained in this way are similar to the recurrence relations of the coefficients of the activity expansions of the correlation functions \cite{penrose1963,minlos-poghosyan1977} inherited from the well-known Kirkwood-Salsburg equation. The principal difference is an additional expression with a M{\"o}bius inversion on the lattice of set partitions, in part inspired by recent work by Dorlas, Rebenko, and Savoie~\cite{dorlas-rebenko-savoie2020}. This corresponds to a version of the Kirkwood-Salsburg equation from which the activity is eliminated; a similar variant of the Kirkwood-Salsburg equation has in fact been derived from the canonical ensemble and applied to prove convergence of density expansions by Bogolyubov et al.~\cite{bogolyubov-hacet1949,bogoljubov-petrina-hacet1969}.

The recurrence relations allow for an inductive proof of an abstract convergence condition, which is our main result (Theorem~\ref{thm:main}). The convergence condition is similar to the convergence condition for activity expansions based on Kirkwood-Salsburg equations given by Bissacot, Fern{\'a}ndez, Procacci \cite{bissacot-fernandez-procacci2010} for discrete polymer systems and Jansen and Kolesnikov \cite{jansen-kolesnikov2020} for continuum Gibbs measures. Known convergence criteria by Lebowitz and Penrose (for non-negative potentials), Groeneveld, and Nguyen and Fern{\'a}ndez are easily recovered, in addition the theorem is applicable to inhomogeneous systems and it yields bounds for the density expansions of all factorial moment measures (also known as correlation functions or distribution functions), thus complementing the results on the pressure and free energy for inhomogeneous systems in \cite{jansen-kuna-tsagkaro2019}.

%

The article is organized as follows. In Section~2 we present our main convergence theorem and explain how to recover from it the convergence criteria by Lebowitz and Penrose~\cite{lebowitz-penrose1964}, Groeneveld~\cite{groeneveld1967c}, and Nguyen and Fern{\'a}ndez \cite{nguyen-fernandez2020}. We also discuss the relation with the work~\cite{jansen-kuna-tsagkaro2019}. The main proof ingredient is a set of recurrence relations for $2$-connected graphs, presented in Section~\ref{sec:recurrence}. Sections~\ref{sec:mainproof} and~\ref{sec:homogeneous} proceed with the proofs of the main convergence theorem and the derived theorems for homogeneous systems in which we recover known convergence criteria. In Section~\ref{sec:activity} we sketch the relation between the activity and density expansions. 

The results and proofs in the main body of the article are phrased in terms of generating functions of weighted labelled graphs, in principle they do not require any knowledge in statistical mechanics. This choice of presentation is motivated by interest in the virial  expansions from a combinatorial point of view \cite{leroux2004,kaouche-leroux2009, faris2010combinatorics, tate2015puzzle}. 
The relation of these generating functions with correlation functions of grand-canonical Gibbs measures is recalled in Appendix~\ref{app:gibbs}. 

\section{Main results} \label{sec:main}

Let $(\mathbb X,\mathcal X)$ be a measurable space and $v$ a pair potential, i.e., a measurable function $v:\mathbb X\times \mathbb X\to \R\cup \{\infty\}$ that is symmetric ($v(x,y) = v(y,x)$ on $\mathbb X^2$). 
Mayer's $f$-function is 
\[
	f(x,y) = \e^{- v(x,y)}  -1. 
\] 
For $n\in\N$, let $\mathcal G_n\supset \mathcal C_n\supset \mathcal D_n$ be the sets of all graphs, connected graphs, and $2$-connected graphs with vertex set $[n] = \{1,\ldots, n\}$. The graphs considered here are undirected and without loops or multiple edges. Remember that a graph is $2$-connected if it is connected and for every vertex $i$, removal of the vertex $i$ and all incident edges results in a graph that is still connected. 

The weight of  a graph $G\in \mathcal G_n$, given $(x_1,\ldots, x_n)\in \mathbb X^n$, is 
\[
	w(G;x_1,\ldots, x_n) := \prod_{\{i,j\}\in E(G)} f(x_i,x_j),
\] 
the product over the empty set is defined to be one. Given a measure $\rho$ on $(\mathbb X,\mathcal X)$, our main goal is to investigate the convergence of the function $\bar d$ on $\mathbb X$ given by 
\[
	\bar d(x_1;\rho):= \sum_{n=1}^\infty \frac{1}{n!}\int_{\mathbb X^n}\Bigl| \sum_{G\in \mathcal D_{n+1}} w(G;x_1,\ldots,x_{n+1})\Bigr| \rho(\dd x_2)\cdots \rho(\dd x_{n+1}). 
\] 
Let us briefly recall a simple convergence condition for $\bar d(x_1;\rho)$, so as to provide some context for Theorem~\ref{thm:main} below. 

\begin{theorem} \label{thm:jkt} \cite{jansen-kuna-tsagkaro2019}
	Suppose that $v\geq 0$ on $\mathbb X^2$ and that there exists a non-negative measurable function $a:\mathbb X\to \R_+$ such that 
	\be \label{eq:suff-jkt}
		\int_{\mathbb X} |f(x,y)|\, \e^{2 a(y)} \rho(\dd y) \leq a(x)
	\ee
	for all $x\in \mathbb X$. Then $\bar d(x;\rho) \leq a(x) <\infty$ on $\mathbb X$. 
\end{theorem}

\noindent Theorem~\ref{thm:jkt} follows from Theorem~3.4 in~\cite{jansen-kuna-tsagkaro2019}. The latter has a  more general convergence condition and covers attractive potentials as well. 

Our treatment of the virial expansion requires the introduction of another set of graphs. Let $W$ and $B$ be two finite disjoint sets. The set $W$ is non-empty but $B$ is allowed to be empty.
We call elements of $W$ white and elements of $B$ black, and define a set of graphs $\mathcal D(W,B)$ with vertex set $W\cup B$ as follows. If $\#W\geq 2$, then a graph $G$ is in $\mathcal D(W,B)$ if and only if: 
\begin{enumerate}[(i) ]
	\item Every black vertex is connected to at least one white vertex by a path in $G$. 
	\item The previous property survives the removal of any vertex (black or white). 
\end{enumerate}  
If $W=\{w\}$ is  a singleton, we define $\mathcal D(\{w\},B):= \varnothing$ if $B \neq \varnothing$ and $\mathcal D(\{w\},\varnothing) =\{G\}$ with $\{G\}$ the graph with unique  vertex $w$ and no edge. In Appendix~\ref{app:gibbs} we recall that these graphs are associated with the density expansions of the correlation functions in statistical mechanics, see also~\cite[Section 6]{stell1964}. A combinatorial motivation is provided by the following lemma.

\begin{lemma} \label{lem:combi2conn}
	Let $G\in \mathcal G_n$ be a graph with at least $n\geq 3$ vertices. Let $W$ be the set of vertices adjacent to $1$, denote $B := \{2,\ldots, n\}\setminus W$, and $G'$ the graph obtained from $G$ by removing $1$ and all incident edges. Then $G$ is $2$-connected if and only if $G'$ is connected and $G'$ is in $\mathcal D(W,B)$ with $\#W\geq 2$. 
\end{lemma} 

\noindent The elementary proof is left to the reader. Graphs in $\mathcal D(W,B)$ are connected when there are exactly two white vertices but need not be connected when there are three white vertices or more. 

Define 
\be \label{eq:psidef}
	\psi\bigl((x_i)_{i\in I}; (x_j)_{j\in J}\bigr):= \sum_{G\in \mathcal D(I,J)} w\bigl( G;(x_i)_{i\in I\cup J}\bigr).
\ee
Notice 
\be \label{eq:psi0}
	\psi(x_1,\ldots, x_s;\varnothing) =  \sum_{G\in \mathcal D([s];\varnothing)} w(G;x_1,\ldots, x_s) = \prod_{1\leq i< j \leq s} \bigl(1+ f(x_i,x_j)\bigr).
\ee
Further define 
\[
	\bar g (x_1,\ldots, x_s;\rho): = \sum_{n=0}^\infty \frac{1}{n!}\int_{\mathbb X^n} \bigl| \psi(x_1,\ldots,x_s;x_{s+1},\ldots, x_{s+n})\bigr| \rho(\dd x_{s+1}) \cdots \rho(\dd x_{s+n}). 
\] 
The contribution from $n=0$ is $|\psi(x_1,\ldots, x_s;\varnothing)|$. 

We view $\bar g(\cdot;\rho)$ as a function from $\sqcup_{s\in \N}\mathbb X^s$ to $\R_+\cup \{\infty\}$. The disjoint union $\sqcup_{s\in \N} \mathbb X^s$ is equipped with the direct sum of the product $\sigma$-algebras $\mathcal X^s$, the function $\bar g(\cdot;\rho)$ is measurable. 

For a concise formulation of our main convergence theorem, we introduce a nonlinear operator on non-negative functions. For $m:\sqcup_{s\in \N} \mathbb X^s\to \R_+$ a measurable function, define a new function $\mathbf K_\rho m:\mathbb X^s \to \R_+\cup \{\infty\}$ as follows. 
For $x_1\in \mathbb X$, set
\[
	(\mathbf K_\rho m)(x_1) = 1
\] 
and introduce the auxiliary function 
\be \label{eq:Rdef} 
	(R_\rho m)(x_1) = \sum_{j=1}^\infty \frac{1}{j!}\int_{\mathbb X^j}\prod_{i=1}^j  |f(x_1,y_j)| m(y_1,\ldots, y_j) \rho^j(\dd \vect y).
\ee
For $s \geq 2$ and $(x_1,\ldots, x_s) \in \mathbb X^s$, set 
\begin{multline*} 
	(\mathbf K_\rho m)(x_1,\ldots,x_s) 
	= \prod_{i=2}^s \bigl( 1+ f(x_1,x_i)\bigr) \Biggl( m(x_2,\ldots,x_s) \\
	+ \sum_{j=1}^\infty \frac{1}{j!}\int_{\mathbb X^j} \prod_{i=1}^j |f(x_1,y_i)| m(x_2,\ldots,x_s,y_1,\ldots, y_j) \rho^j(\dd \vect y)\Biggr)\\
	\times \Biggl( \sum_{k=0}^\infty  \Bigl( (R_\rho m)(x_1) \Bigr)^k \Biggr).
\end{multline*} 
The fixed-point equation $\mathbf K_\rho m = m$ is closely related to the Kirkwood-Salsburg equation, see Eq.~\eqref{eq:ks} in Appendix~\ref{app:gibbs}. Our main theorem is the following abstract convergence condition. 

\begin{theorem} \label{thm:main}
	Suppose there exists a measurable function $m:\sqcup_{s\in \N} \mathbb X^s\to \R^+$ with $(R_\rho m)(x_1) <1$ for all $x_1\in \mathbb X$ and 
	\be \label{eq:suff1}
		(\mathbf K_\rho m)(x_1,\ldots, x_s) \leq m (x_1,\ldots, x_s)
	\ee
	for all $s\in \N$ and $(x_1,\ldots, x_s) \in \mathbb X^s$. Then 
	\[
		\bar g(x_1,\ldots, x_s;\rho) \leq m(x_1,\ldots, x_s) < \infty
	\] 
	for all $s\in \N$ and $(x_1,\ldots, x_s) \in \mathbb X^s$ 
	and 
	\[
		\bar d(x_1;\rho) \leq - \log \Bigl( 1 - (R_\rho m)(x_1)\Bigr) < \infty 
	\] 
	for all $x_1\in \mathbb X$. 
\end{theorem} 

\noindent The theorem is inspired by a similar convergence condition for activity expansion~\cite{jansen-kolesnikov2020}, see Theorem~\ref{thm:activity} below. Theorem~\ref{thm:main} is proven in Section~\ref{sec:mainproof}. 

Before we discuss the relation of Theorem~\ref{thm:main} to Theorem~\ref{thm:jkt}, let us explain how to recover some known convergence criteria.  We specialize to $\mathbb X = \R^d$ with the Borel $\sigma$-algebra and translationally invariant pair potentials $v(x,y) = v(0,y-x)$ which we write as $v(y-x)$ instead. We assume that the pair potential $v$ is stable, i.e., there exists $B\in (0,\infty)$ such that 
\[
	\sum_{1\leq i <j \leq n} v(x_i-x_j) \geq - n B
\] 
for all $n\geq 2$ and $x_1,\ldots, x_n\in \R^d$. Further assume that 
\[
	C:= \int_{\R^d}\bigl|\e^{- v(y)} - 1\bigr|\dd y 
\] 
is finite. Let $\lambda$ be the Lebesgue measure on $\R^d$. We work with translationally invariant measures $\rho \lambda(\dd x)$ where $\rho>0$ and write $\bar g(\vect x;\rho)$ and $\bar d(x_1;\rho)$ instead of $\bar g(\vect x;\rho \lambda)$ and  $\bar d(x_1;\rho\lambda)$. The following theorems are proven in Section~\ref{sec:homogeneous} by specializing Theorem~\ref{thm:main} to different choices of weight functions $m(x_1,\ldots, x_s)$. 

\begin{theorem} \label{thm:LP}
Let $v$ be a stable translationally invariant pair potential in $\R^d$ and $\rho>0$. 
	Set $u:= \e^{2 B}$. Assume there exists $\kappa \geq 0$ such that 
	\be \label{eq:lpsuff} 
		 u\e^{C\rho \kappa} \leq \kappa (2 - \e^{C\rho \kappa}).
	\ee
	Then for all $s\geq 2$ and $x_1,\ldots, x_s\in \R^d$, 
	\[
		\bar g(x_1,\ldots, x_s;\rho) \leq \kappa^{s}
	\]
	and 
	\[
		\bar d(x_1;\rho) \leq - \log \Bigl(2-\e^{C\rho\kappa} \Bigr)<\infty.
	\]
\end{theorem} 

\noindent Notice that condition~\eqref{eq:lpsuff} is equivalent to (think $\mu = C\rho \kappa$)
\[
	 \rho \leq \frac{1}{C u} \sup_{\mu>0} \mu(2 \e^{-\mu} - 1) \simeq \frac{0.14477}{Cu}.
\] 
For $u=1$ (non-negative potentials), the right-hand side is precisely the radius of convergence given by Lebowitz and Penrose \cite{lebowitz-penrose1964}.  For $u >1$ ($B>0$), the bound given here is less good than the bound $\sup_{\mu>0} (1+u) \e^{-\mu} - 1)\mu/(Cu^2)$ given in~\cite{lebowitz-penrose1964}. The density expansions of correlation functions are treated as well by Lebowitz and Penrose, see \cite[Section 7]{lebowitz-penrose1964}.

\begin{theorem}[Groeneveld recovered] \label{thm:groe}
Let $v$ be a stable translationally invariant pair potential in $\R^d$ and $\rho>0$. 
	Set $u:= \e^{2 B}$. Assume there exists $\kappa \geq 0$ such that 
	\be \label{eq:groesuff} 
		(1+ u)\e^{C\rho \kappa} - 1\leq \kappa.
	\ee
	Then for all $s\geq 2$ and $x_1,\ldots, x_s\in \R^d$, 
	\[
		\bar g(x_1,\ldots, x_s;\rho) \leq \kappa^{s-1}
	\] 
	and 
	\[
		\bar d(x_1;\rho) \leq - \log \Bigl(1-\frac1\kappa \bigl(\e^{C\rho\kappa}-1\bigr) \Bigr)<\infty.
	\] 	
\end{theorem} 

\noindent Notice that condition~\eqref{eq:groesuff} is equivalent to (think $\mu = C\rho \kappa$)
\[
	\rho \leq \frac{1}{C} \sup_{\mu>0}\frac{\mu}{(1+u)\e^\mu - 1} =:R_{\mathrm{Groe}}.
\] 
The right-hand side is equal to the radius of convergence given by Groeneveld \cite[\S 4, Corollary 2]{groeneveld1967c}. The bounds for $\bar g$ and $\bar d$ are not explicitly stated in \cite{groeneveld1967c} but should follow from the considerations therein. For non-negative potentials, we may take $B=0$, $u=1$, and a numerical evaluation yields $R_{\mathrm{Groe}} \simeq 0.23196/C$. 
(For non-negative potentials, Groeneveld also states a better bound  $0.27846/C$, but his work~\cite{groeneveld1967c} does not include a proof and the forthcoming work that should have had the proof does not seem to have been published.)

Next we further specialize to non-negative pair potentials. For $\mu>0$, define
\[
	\Psi(\mu):= 1+ \sum_{k=1}^\infty \frac{\mu^k}{k!} \int_{(\mathbb R^d)^k} \prod_{i=1}^k |f(0,y_i)|\prod_{1\leq i < j \leq k}\bigl( 1+ f(y_i,y_j)\bigr) \dd \vect y. 
\] 

\begin{theorem}[Nguyen-Fern{\'a}ndez recovered] \label{thm:NF}
	Let $v$ be a non-negative, translationally invariant pair potential in $\R^d$, and $\rho>0$. 
	Suppose there exists a scalar $\kappa>0$ such that 
	\be \label{eq:NF}
		2\Psi(\rho\kappa) - 1 \leq \kappa. 
	\ee
	Then 	for all $s\in \N$ and $(x_1,\ldots, x_s) \in \mathbb X^s$, we have 
	\[
		\bar g(x_1,\ldots, x_s;\rho) \leq  \kappa^{s-1} \prod_{1\leq i < j \leq s}\bigl( 1+ f(x_i,x_j)\bigr) 
	\] 
	and 
	\[
		\bar d(x_1;\rho) \leq - \log \Bigl( 1- \frac1\kappa \bigl(\Psi(\rho\kappa) - 1\bigr)\Bigr) < \infty. 
	\]
\end{theorem} 

\noindent Notice that condition~\eqref{eq:NF} is equivalent to (think $\mu = \rho \kappa$) 
\[
	\rho \leq \sup_{\mu>0} \frac{\mu}{2\Psi(\mu) -1} =:R_{\mathrm{NF}}.
\] 
The right-hand side is precisely the lower bound to the radius of convergence of the virial series proven by Nguyen and Fern{\'a}ndez~\cite[Corollary~4.6]{nguyen-fernandez2020}. For hard spheres in dimension $d=2$, a numerical evaluation yields $R_{\mathrm{NF}}\simeq 0.300224 /C$ \cite[Application~4.9]{nguyen-fernandez2020}. The precise bound on $\bar g$ in Theorem~\ref{thm:NF} is new. 

Turning back to general measure spaces $(\mathbb X,\mathcal X)$, let us address the relation between Theorem~\ref{thm:jkt} and Theorem~\ref{thm:main}. To that aim we first derive a simpler convergence condition, for non-negative pair potentials. We make the ansatz 
\be \label{eq:simplem}
	m(x_1,\ldots, x_s) = \e^{b(x_1)+\cdots + b(x_s)}
\ee
for some measurable function $b:\mathbb X\to \R_+$.
Then 
\[
	(R_\rho m)(x_1) = \exp\Bigl( \int_\mathbb X|f(x,y)|\e^{b(y)} \rho(\dd y)\Bigr) -1.
\] 
Using $1+ f\leq 1$, it is easily checked that if 
\be \label{eq:suff2}
   \frac{\exp(\int_\mathbb X|f(x,y)|\e^{b(y)} \rho(\dd y))}{2 - \exp(\int_\mathbb X|f(x,y)|\e^{b(y)} \rho(\dd y))}\leq \e^{b(x)}
\ee
for all $x\in \mathbb X$, with strictly positive denominator, then the sufficient convergence condition~\eqref{eq:suff1} is met. Because of 
\[
	\frac{1}{2 - \e^t}  = \frac{1}{1 - (\e^t - 1)}\geq 1+ (\e^t -1) = \e^t
\] 
for all $t\in [0,\log 2)$, condition~\eqref{eq:suff2} implies 
\be
	\exp\Bigl( 2 \int_\mathbb X|f(x,y)| \e^{b(y)}\rho(\dd y)\Bigr) \leq \e^{b(x)}
\ee
hence 
\[
	\int_{\mathbb X} |f(x,y)| \e^{b(y)} \rho(\dd y) \leq \frac{b(x)}{2}
\]
and $\rho$ satisfies condition~\eqref{eq:suff-jkt} from Theorem~\ref{thm:jkt} with $a(x) = b(x)/2$. 

Put differently, Theorem~\ref{thm:main} applied to the simplest weight function~\eqref{eq:simplem} yields a convergence condition that is \emph{worse} than the condition~\eqref{eq:suff-jkt} from \cite{jansen-kuna-tsagkaro2019}. Condition~\eqref{eq:suff2} however is similar to the Lebowitz-Penrose condition from Theorem~\ref{thm:LP}. Theorems~\ref{thm:groe} and~\ref{thm:NF} suggest that criteria better than Theorem~\ref{thm:jkt} might be obtained from our Theorem~\ref{thm:main}  with other weight functions $m$ also in the inhomogeneous case, but this is beyond the purpose of the present work. 

\section{Recurrence relations for $2$-connected graphs} \label{sec:recurrence}

Remember the set $\mathcal D_n$ of $2$-connected graphs with vertex set $[n]=\{1,2,\ldots,n\}$. The graph $G$ with vertices $1,2$ and edge $\{1,2\}$ is considered $2$-connected. 
For $n\geq 2$, let
\[
	D_{n}(x_1,x_2,\ldots, x_{n}) = \sum_{G\in \mathcal D_{n}} w(G;x_1,\ldots,x_n).
\] 
Further remember the weights $\psi((x_i)_{i\in I}, (x_j)_{j\in J})$ from~\eqref{eq:psidef} and~\eqref{eq:psi0} and the graphs $\mathcal D(I,J)$ defined in Section~\ref{sec:main}. To lighten notation, when there is no risk of confusion we use the shorthand 
\[
	\psi(I, J):= \psi((x_i)_{i\in I}, (x_j)_{j\in J}),\quad D(I) = D_{\#I}( (x_i)_{i\in I}). 
\] 
Let $\mathcal P(V)$ be the collection of set partitions $\{V_1,\ldots, V_m\}$ of $V$ (non-empty disjoint sets with union $V$, order irrelevant). 

\begin{prop} \label{prop:2conn}
	For all $n\geq 2$ and $(x_1,\ldots, x_n) \in \mathbb X^n$,
	\begin{multline*}
		D(\{1,\ldots, n\}) = \sum_{L\subset \{2,\ldots,n\}} \prod_{\ell \in L} f(x_1,x_\ell) \sum_{m=1}^n (-1)^{m-1} (m-1)! \\
			\times\sum_{\{V_1,\ldots,V_m\}\in\mathcal P(\{2,\ldots,n\})} \prod_{r=1}^m \psi\bigl( L \cap V_r, (\{2,\ldots,n\}\setminus L)\cap V_r\bigr).
	\end{multline*} 
\end{prop} 

\begin{proof} 
	If $n =2$, then $D(\{1,2\}) = f(x_1,x_2)$ and the proposition is trivial. 
	Consider $n\geq 3$. Let $G\in \mathcal G_n$, $L\subset \{2,\ldots,n\}$ the set of vertices adjacent to $1$, and $G'$ the graph with vertex set $\{2,\ldots, n\}$ obtained from $G$ by deleting the vertex $1$ and all incident edges. By Lemma~\ref{lem:combi2conn}, $G$ is in $\mathcal D_n$ if and only if $G'$ is in $\mathcal D(L, \{2,\ldots, n\}\setminus L)$ and $G'$ is connected. Define 
	\[
		\psi_c(I,J) = \sum_{\substack{G\in \mathcal D(I,J)\\ G\text{ connected}}} w(G;(x_i)_{i\in I\cup J}),
	\]
	then 
	\be \label{eq:D1}
		D(\{1,\ldots,n\}) = \sum_{L\subset \{2,\ldots,n\}} \prod_{\ell \in L} f(x_1,x_\ell) \psi_c\bigl(L,\{2,\ldots,n\}\setminus L\bigr).
	\ee
	On the other hand, for every $J\subset \{2,\ldots, n\}$,  summing over connected component of graphs $G\in \psi(L,J)$, one finds 
	\[
		\psi(L,J) = \sum_{m\geq 1}\ \sum_{\{V_1,\ldots, V_m\}\in \mathcal P(L\cup J)}\ \prod_{r=1}^m \psi_c( L\cap V_r,J\cap V_r)
	\] 
	where  we define $\psi(\varnothing, B) =0$ for all $B$. Define $b_L(V):= \psi(L\cap V, V\setminus L)$ and $a_L(V) = \psi_c(L\cap V,V\setminus L)$, 
	then we may rewrite the previous equation as 
	\[
		b_L( L\cup J) = \sum_{m\geq 1}\ \sum_{\{V_1,\ldots, V_m\}\in \mathcal P(L\cup J)} a_L(V_1)\cdots a_L(V_m).
	\] 
	A standard formula for the M{\"o}bius inversion on the lattice of set partitions \cite[Section 6.2]{malyshev-minlos-book} then yields 
	\[
		a_L(L\cup J) = \sum_{m\geq 1} (-1)^{m-1} (m-1)! \sum_{\{V_1,\ldots, V_m\}\in \mathcal P(L\cup J)} b_L(V_1)\cdots b_L(V_m)
	\] 
	which gives 
	\be \label{eq:D2}
		\psi_c( L,J) = \sum_{m\geq 1}(-1)^{m-1} (m-1)! \sum_{\{V_1,\ldots, V_m\}\in \mathcal P(L\cup J)} \prod_{r=1}^m \psi( L\cap V_r,J\cap V_r).
	\ee
	The proposition follows by inserting~\eqref{eq:D2} into~\eqref{eq:D1}.
\end{proof} 

\begin{prop} \label{prop:recursion}  
 	Let $I$ and $J$ be two non-empty disjoint finite sets with $\#I \geq 2$, and $x_i$, $i\in I\cup J$,  elements in $\mathbb X$. Pick $\iota(I) \in I$ and set $I':= I\setminus \iota(I)$. Then 
 	\begin{multline*}
 		\psi(I;J) =  \prod_{i\in I'} \bigl(1+ f(x_{\iota(I)}, x_i)\bigr) \Biggl( \psi(I';J) + \sum_{\substack{L\subset J:\\ L \neq \varnothing}} \prod_{\ell \in L} f(x_{\iota (I)}, x_\ell)\\
 			\times  \sum_{m\geq 1} (-1)^{m-1} \sum_{\substack{(L_1,\ldots,L_m)\\ (J_1,\ldots, J_m)}} \psi(I'\cup L_1;J_1)\psi(L_2;J_2)\cdots \psi(L_m;J_m)\Biggr),
 	\end{multline*}
 	where the last sum is over tuples such that 
 	\begin{itemize} 
 		\item $L_1,\ldots, L_m$ are pairwise disjoint with $L_1$ possibly empty but $L_2,\ldots, L_m$ non-empty, and $L = L_1\cup \cdots\cup L_m$. 
 		\item $J_1,\ldots, J_m$ are pairwise disjoint with $J_1\cup \cdots\cup J_m = J\setminus L$. Each $J_i$ is allowed to be empty.
 	\end{itemize} 
\end{prop} 

The proposition is analogous to \cite[Eq.~(5.9)]{penrose1963},  \cite[Lemma 6.2]{poghosyan-ueltschi2009} and~\cite[Lemma 4.1]{jansen2019clustergibbs} which give a recurrence relation for the coefficients of the activity expansion of the correlation functions.

\begin{proof} 
	We prove first a simpler formula with a different weight function $\widehat \psi$ (Eq.~\eqref{eq:rekey1}), and then we express $\widehat \psi$ with the help of an inverse M{\"o}bius transform (Eq.~\eqref{eq:keyrec4}). 
	
   Let $G\in \mathcal D(I,J)$ and $L\subset J$ the set of black vertices incident to the white vertex $1$. Removing the vertex $1$ and all incident vertices from $G$ results in a graph  $G' \in \mathcal D(I'\cup L, J\setminus L)$. 
 The graph $G'$ necessarily has the following property: every black vertex in $J$ connects to $I'$. Indeed, if it didn't, then in $G$ every path linking the black vertex to a white vertex would pass through $1$ hence removal of the white vertex $1$ would destroy the defining property (i) of $\mathcal D(I,J)$, contradicting $G\in \mathcal D(I,J)$. Thus we have checked that if $G$ is in $\mathcal D(I,J)$, then $G'$ is in the set $\widehat{ \mathcal D}(I';L;J\setminus L)$ given by 
\[
	\widehat{ \mathcal D}(I';L;J\setminus L) 
	  = \bigl\{ G' \in \mathcal D(I'\cup L;J\setminus L):\, \text{every }j\in J\text{ connects to } I'\}.
\] 
Conversely, if $G'$ is a graph in $\widehat{ \mathcal D}(I';L;J\setminus L)$, then adding the vertex $\iota(I)$, the edges $\{1,\ldots, \ell\}$ and none, some, or all of the edges $\{\iota(I),i\}$ with $i\in I'$, gives a graph in $G'$. Define 
\[
	\widehat \psi(I', L, J\setminus L) = \sum_{G'\in\widehat{ \mathcal D}(I', L, J\setminus L)} w\bigl( G; \vect x_{I'\cup L \cup J}\bigr).  
\] 
The considerations above yield
\be \label{eq:rekey1}
	 		\psi(I;J) =  \prod_{i\in I'} \bigl(1+ f(x_{\iota(I)}, x_i)\bigr) \sum_{L\subset J} \prod_{\ell \in L} f(x_{\iota (I)}, x_\ell)	\widehat \psi(I', L, J\setminus L),
\ee
the contribution from $L = \varnothing$ is $\widehat \psi(I',J)$. 

Next we express the $\widehat \psi$-weights for $L\neq \varnothing$ in terms of the $\psi$-weights with an inverse M{\"o}bius transform. To that aim, keeping $I'$ fixed, we define a new set of graphs that corresponds, roughly, to the connected components of graphs $G\in \mathcal D(I', L, J\setminus L)$. For $W,B$ disjoint finite sets (white and black vertices), define
\[
	\mathcal D_c(W,B):= \begin{cases} 
					\widehat{\mathcal D}(I',W\setminus I',B), &\quad I'\subset W,\\
					\{G'\in \mathcal D(W,B):\ G'\text{ is connected}\}, &\quad I'\cap W = \varnothing,\\
					\varnothing, &\quad \text{else}.
			\end{cases} 
\] 
The case distinction in the definition of $\mathcal D_c(W,J)$ is only needed for $\#I'\geq 2$; when $\#I' =1$ and $W\supset I'$, a graph $G'\in \mathcal D(W,B)$ is in $ \widehat{\mathcal D}(I',W\setminus I', B)$ if and only if it is connected. Further set 
\[
	\psi_c(W,B) = \sum_{G'\in \mathcal D_c(W,B)} w(G;\vect x_{W\cup B})
\] 
and notice that for all $W\supset I'$, 
\be \label{eq:rekey10} 
	\psi_c(W,B) = \widehat \psi(I',W\setminus I', B).
\ee
Consider first the simpler case when $I'$ has cardinality $1$. Let $W$ and $B$ be two finite disjoint sets and $G'$ a graph in $\mathcal D(W,B)$. The graph $G'$ splits into connected components $G'_1,\ldots, G'_m$ with vertex sets $V_1,\ldots, V_m$. Notice $V_1\cup \cdots \cup V_m = W\cup B$. 
Set 
 \[
 	M_k:= V_k\cap M,\quad B_k := V_k \cap B\quad (k=1,\ldots,m).  
 \]
It is easily checked that each $G'_k$ is in $\mathcal D(W_k,B_k)$ for all $k$. In addition, each $B_k$ is non-empty. Thus we find 
\be \label{eq:rekey2}
	\psi(W,B) = \sum_{m=1}^{\#(W\cup B)} \sum_{\{V_1,\ldots, V_m\}} \prod_{k=1}^m \psi_c(V_k\cap W, V_k\cap B)
\ee 
where the sum is over set partitions of $W\cup B$ and we set $\psi(W',B') = \varnothing$ if $W' = \varnothing$. Using the well-known formula for the inverse M{\"o}bius transform on the lattice of set partitions, see e.~g.\ \cite[Chapter 2.6]{malyshev-minlos-book}, we deduce
\be \label{eq:rekey3}
	\psi_c(W,B) = \sum_{m=1}^{\#(W\cup B)} (-1)^{m-1} (m-1)! \sum_{\{V_1,\ldots, V_m\}} \prod_{k=1}^m \psi(V_k\cap W, V_k\cap B)
\ee
with summation over set partitions of $W\cup B$. We apply the formula to $W=I'\cup L$ and $B=J\setminus L$ with disjoint $I'$ and $J$, and $L\subset J$. In view of~\eqref{eq:rekey10}, the left side of~\eqref{eq:rekey2} is $\widehat \psi(I',L,J\setminus L)$. On the right side of~\eqref{eq:rekey10}, we get rid of the factorial $(m-1)!$ by ordering the blocks of the partitions. The uniquely defined block containing the singleton $I'$ is labelled $V_1$. There are $(m-1)!$ ways to order the remaining blocks, we obtain 
\begin{multline*}
	\widehat \psi(I',L,J\setminus L) = \sum_{m\geq 1} (-1)^{m-1} \\
		\times   \sum_{(V_1,\ldots, V_m)} \psi(V_1\cap (I'\cup L), V_1\cap (J\setminus L)) \prod_{k=2}^m \psi\bigl(V_k \cap L, V_k \cap (J\setminus L)\bigr)
\end{multline*} 
with summation over \emph{ordered} set partitions of $W\cup  B= I'\cup J$ such that $V'_1\supset I'$. 
The correspondence $L_k = V_k \cap L$, $J_k = V_k \cap J$ allows us to rewrite the previous equation as 
\be \label{eq:keyrec4}
	\widehat \psi(I',L,J\setminus L) = \sum_{m\geq 1} (-1)^{m-1} 
		\sum_{\substack{(L_1,\ldots, L_m)\\(J_1,\ldots, J_m)}} \widehat \psi(I'\cup L_1, J_1) \prod_{k=2}^m \psi(L_k,J_k)
\ee
with the domain of summation formulated as in the lemma.

A similar reasoning shows that~\eqref{eq:keyrec4} stays valid when $\#I\geq 2$. Instead of connected components of graphs $G'$, one needs to look at connected components of the graph obtained by identifying all vertices of $I'$. The details are left to the reader. 

The proposition follows by inserting~\eqref{eq:keyrec4} into~\eqref{eq:rekey1}. 
\end{proof} 

\section{Inductive proof of convergence} \label{sec:mainproof}

The recurrence relation from Proposition~\ref{prop:recursion} allows for an inductive proof of Theorem~\ref{thm:main}. The induction is over the maximum number $N$ of black vertices, as in~\cite{ueltschi2004}. Let
\begin{multline} \label{eq:hbardef}
	\bar g^{(N)}_\rho( x_1,\ldots, x_s):= |\psi(x_1,\ldots,x_s;\varnothing)| \\
	+ 	\sum_{n=1}^N \frac{1}{n!}\int_{\mathbb X^n} \bigl|\psi(x_1,\ldots, x_s; y_{s+1},\ldots, y_{s+n})\bigr| \rho^n(\dd \vect y).
\end{multline} 
Notice 
\[
	\bar g_\rho^{(0)} (x_1,\ldots,x_s) = \prod_{1\leq i<j\leq s}\bigl( 1+ f(x_i,x_j)\bigr)
\] 
because of~\eqref{eq:psi0}, moreover 
\[
	\bar g_\rho^{(N)}(x_1) =1\quad (N\in \N).
\] 
The recursive relation from Proposition~\ref{prop:recursion} translates into a recursive inequality for the partial sums. 

\begin{prop} \label{prop:fpi}
For all $N\in \N$, $s\geq 2$, and every $\iota\in \{1,\ldots,s\}$, 
\begin{multline} \label{eq:barkirk} 
	\bar g_{\rho}^{(N+1)} (x_1,\ldots,x_s) 
 \leq  \prod_{\substack{1\leq j\leq s:\\ j \neq \iota}} \bigl( 1+ f(x_\iota,x_j)\bigr)\Biggl( \bar g_{\rho}^{(N)} ( (x_j)_{j\neq\iota} ) \\
		    + \sum_{k=1}^\infty \frac{1}{k!}\int_{\mathbb X^k}\prod_{i=1}^k |f(x_\iota,y_i)| \bar g_{\rho}^{(N)} ( (x_j)_{j\neq \iota}, y_1,\ldots, y_k)\rho^k(\dd \vect y)\Biggr)  \\
		\times \Biggl\{ 1 - \sum_{k=1}^\infty \frac{1}{k!}\int_{\mathbb X^k}\prod_{i=1}^k |f(x_\iota,y_i)| \bar g_{\rho}^{(N)} (y_1,\ldots, y_k)\rho^j(\dd \vect y)\Biggr\}^{-1}
\end{multline} 
with the convention $1/(1-q)=\infty$ when $q\geq 1$. 
\end{prop}

\noindent In particular, for $\iota =1$ we obtain
\be \label{eq:horeka}
	\bar g^{(N+1)}_\rho \leq \mathbf K_\rho \bar g_\rho^{(N)}. 
\ee
The proposition is similar to \cite[Proposition 4.1]{jansen2019clustergibbs}.

\begin{proof} 
	To lighten notation we drop the $\rho$-index and write down the proof for $\iota =1$ only, the proof in the general case is similar. For $\iota =1$ the right-hand side of~\eqref{eq:barkirk} is equal to 
\begin{multline} \label{eq:barkirk2} 
   \prod_{j=2}^s (1+ f(x_1,x_j)) \Biggl\{ \bar g^{(N)} (x_2,\ldots,x_s) \\
		    + \sum_{r=1}^\infty \frac{1}{r!} \int_{\mathbb X^r} \prod_{i=1}^r |f(x_1, y_i)|\\
		    	\times  \Biggl( \sum_{m\geq 1} \sum_{(L_1,\ldots, L_m)}  \bar g^{(N)} (x_2,\ldots,x_s, \vect y_{L_1}) \prod_{j=2}^m \bar g^{(N)} (\vect y_{L_j})\Biggr) \rho^r(\dd \vect y)\Biggr\}
\end{multline} 
where the sum is over tuples $(L_1,\ldots, L_m)$ of pairwise disjoint sets with union $\{1,\ldots, r\}$ and possibly empty $L_1$. When $m=1$ we set $\prod_{j=2}^m = 1$. Let us define 
\[
	\psi^{(N)}(I;J) = \begin{cases} 
			\psi(I;J), &\quad \#J \leq N,\\
			0, &\quad \text{else}.
		\end{cases} 
\] 
Replacing $\psi$ with $\psi^{(N)}$ on the right-hand side of the definition~\eqref{eq:hbardef} allows us to extend the summation from $n\leq N$ to $n\in \N_0$. We insert this representation of $\bar g_s^{(N)}$ into~\eqref{eq:barkirk2} and see that the right-hand side of~\eqref{eq:barkirk} is of the form
\be \label{eq:gaux}
	g_0\bigl ( (x_j)_{j\in I'}\bigr) + \sum_{n=1}^\infty \frac{1}{n!}\int_{\mathbb X^n} g_n( (x_j)_{j\in I'} ; x_{s+1},\ldots, x_{s+n}) \rho(\dd x_{s+1}) \cdots \rho(\dd x_{s+n})
\ee
with $I' = \{2,\ldots, s\}$, 
\[
	g_0\bigl ( (x_j)_{j\in I'}\bigr) =  \prod_{i\in I'}(1+ f(x_1,x_i)) |\psi^{(N)}(I';\varnothing)|,
\] 
and 
\begin{multline*} 
	g_n( (x_j)_{j\in I'} ; x_{s+1},\ldots, x_{s+n}) =  \e^{ - W(x_\iota;(x_j)_{j\in I'})}\Biggl( |\psi^{(N)}(I';\{s+1,\ldots, s+n\})|\\
	+ \sum_{\substack{L\subset [n]:\\ L \neq \varnothing}} \prod_{\ell\in L} |f(x_\iota, y_\ell)|
	\sum_{m\geq 1} \sum_{\substack{(L_1,\ldots,L_m),\\ (J_1,\ldots, J_m)}} |\psi^{(N)} \bigl( I'\cup L_1;J_1)|\cdots |\psi^{(N)}(L_m;J_m)|\Biggr)
\end{multline*} 
where $(L_1,\ldots, L_m)$ and $(J_1,\ldots, J_m)$ pairwise disjoint, $L_2,\ldots,L_m$ non-empty, and $\cup_i L_i = L$, $\cup_i J_i = \{s+1,\ldots, s+n\}\setminus L$. Proposition~\ref{prop:recursion} yields 
\begin{align*} 
	g_n( (x_j)_{j\in I'} ; x_{s+1},\ldots, x_{s+n}) &\geq |\psi^{(N+1)} (\{1,\ldots, s\}; \{s+1,\ldots, s+n\})|\\
	g_0( (x_j)_{j\in I'}) &\geq |\psi^{(N+1)} (\{1,\ldots, s\};\varnothing)|.
\end{align*} 
It follows that the expression~\eqref{eq:gaux} is larger or equal to $\bar g^{(N+1)}(x_1,\ldots, x_s)$. 
We have already observed that the right-hand side of~\eqref{eq:barkirk} is larger or equal to~\eqref{eq:gaux}. The proof is complete.  
\end{proof} 

\noindent The relation for $2$-connected graphs from Proposition~\ref{prop:2conn} allows us to relate $\bar d(x_1;\rho)$ and $\bar g(x_1,\ldots, x_s;\rho)$. 

\begin{prop} \label{prop:2conn2}
	We have 
	\[
		\bar d(x_1;\rho) \leq - \log \Biggl( 1- \sum_{j=1}^\infty \frac{1}{j!}\int_{\mathbb X^j}\prod_{i=1}^j |f(x_1,y_j)| \bar g(y_1,\ldots, y_j) \rho^j(\dd \vect y)\Biggr)
	\] 
	with the convention $- \log (1- q) = \infty$ for $q \geq 1$. 
\end{prop} 

\noindent The proof is similar to the proof of Proposition~\ref{prop:fpi} and it is therefore omitted.

\begin{proof} [Proof of Theorem~\ref{thm:main}]
	Let $m:\sqcup_{s\in \N}\mathbb X^s\to \R_+$ satisfy condition~\eqref{eq:suff1}. Because of $(\mathbf K_\rho m)(x_1) = 1$ by definition of $\mathbf K_\rho$, we have $m(x_1) \geq 1$ for all $x_1\in \mathbb X$. 
	Condition~\eqref{eq:suff1} implies 
  \[
  	m(x_1,\ldots, x_s) \geq \prod_{j=2}^s \bigl( 1+ f(x_1, x_j)\bigr) m(x_2,\ldots, x_s) 
  \] 	
  and a straightforward induction over $s\in \N$ yields 
  \be \label{eq:mprelim}
  	m (x_1,\ldots, x_s) \geq \prod_{1\leq i <j \leq s} \bigl( 1+ f(x_i,x_j)\bigr). 
  \ee
  Next we 
   show by induction over $N\in \N_0$ that the functions $\bar g_\rho^{(N)}$ from~\eqref{eq:hbardef} satisfy 
	\be \label{eq:hindu} 	
	   \bar g_s^{(N)}(x_1,\ldots, x_s) \leq m(x_1,\ldots, x_s).
	\ee	
	 For $N=0$, we have in view of~\eqref{eq:mprelim}
	\[
		\bar g_\rho^{(0)}(x_1,\ldots,x_s) = |\psi(x_1,\ldots,x_s;\varnothing)| = \prod_{1\leq i <j \leq s}\bigl( 1+ f(x_i,x_j)\bigr))\leq m (x_1,\ldots, x_s).  
	\] 
	For the induction step, suppose that~\eqref{eq:hindu} holds true for all $s\in \N$ and some $N\in \N$.
	We apply the recursive inequality from Proposition~\ref{prop:fpi} with $\iota =1$ in the form~\eqref{eq:horeka} and deduce 
	\[
		\bar g_\rho^{(N+1)} \leq \mathbf K_\rho \bar g_\rho^{(N)} \leq \mathbf K_\rho m\leq m. 
	\] 	
	This completes the induction. 
	Passing to the limit $N\to\infty$ in~\eqref{eq:hindu} we obtain the bound for $\bar g(\cdot;\rho)\leq m$.
	The bound on  $\bar d (x_1;\rho)$ then follows from Proposition~\ref{prop:2conn2}. 
\end{proof} 

\section{Application to homogeneous systems} \label{sec:homogeneous}

Because of the stability of the pair potential, for every $s\in \N$ and $(x_1,\ldots, x_s) \in \mathbb X^s$ there exists $i\in \{1,\ldots, s\}$ such that 
	\be \label{eq:ichoice}
		\sum_{j\neq i} v(x_i- x_j) \geq - 2B.
	\ee
 Let $\iota(x_1,\ldots, x_s)$ be the smallest such $i$ and given a map $m:\sqcup \mathbb X^s \to \R_+$, define 
	\[
		(\Pi m)(x_1,\ldots, x_s):= m(x_{i}, x_1,\ldots, \widehat x_i, \ldots, x_s),\quad i = \iota(x_1,\ldots, x_s) 
	\] 
	where $\widehat x_i$ means omission of the variable. Thus $\Pi$ puts the selected index first and leaves the order of the variables otherwise unchanged. The conclusions of Theorem~\ref{thm:main} hold true as well if we can find a weight function such that 
	\be \label{eq:bsuff}
		(R_\rho \Pi m)(x_1) = (R_\rho m)(x_i)< 1, \quad \mathbf K\Pi m \leq m. 
	\ee
	Permuting variables is fairly standard in the context of Kirkwood-Salsburg equations~\cite{ruelle1969book}.

\begin{proof} [ Proof of Theorem~\ref{thm:LP}]
 	In the sufficient condition~\eqref{eq:bsuff} we make the ansatz 
	\[
		m(x_1,\ldots, x_s) := \kappa^{s}.
	\] 
	In order for $m$ to satisfy condition~\eqref{eq:bsuff}, it is enough that  
	\[  
       u \Bigl( \kappa^{s-1} 
		    + \sum_{j=1}^\infty \frac{1}{j!} C^j \rho^j \kappa^{s-1+j} \Bigr) 
		\Bigl\{ 1 - \sum_{j=1}^\infty \frac{1}{j!} C^j \rho^j \kappa^{j} \Bigr\}^{-1}
		\leq \kappa^{s}
	\]
	and 
	\[
		\sum_{j=1}^\infty \frac{1}{j!} C^j \kappa^{j} \rho^j <1.
	\] 	
	Equivalently, 
	\[
		\e^{C\rho \kappa}-1<1,\quad u \frac{\e^{C\rho \kappa} }{1- (\e^{C\rho \kappa}-1)} \leq \kappa
	\] 
	which is indeed satisfied under the conditions of the theorem. Therefore condition~\eqref{eq:bsuff} and the conclusions of Theorem~\ref{thm:main} hold true. Theorem~\ref{thm:LP} follows. 
\end{proof}

\begin{proof} [ Proof of Theorem~\ref{thm:groe}]
 	In the sufficient condition~\eqref{eq:bsuff} we make the ansatz 
	\[
		m(x_1,\ldots, x_s) := \kappa^{s-1}.
	\] 
	In order for $m$ to satisfy condition~\eqref{eq:bsuff}, it is enough that 
	\[
       u \Bigl( \kappa^{s-2} 
		    + \sum_{j=1}^\infty \frac{1}{j!} C^j \rho^j \kappa^{s-2+j} \Bigr) 
		\Bigl\{ 1 - \sum_{j=1}^\infty \frac{1}{j!} \rho^j \kappa^{j-1} \Bigr\}^{-1}
		\leq \kappa^{s-1}
	\]
	and 
	\[
		\sum_{j=1}^\infty \frac{1}{j!} C^j \kappa^{j-1} \rho^j <1.
	\] 	
	Equivalently, 
	\[
		\frac1\kappa(\e^{C\rho \kappa}-1)<1,\quad u \frac{\e^{C\rho \kappa} }{1- \frac1\kappa(\e^{C\rho \kappa}-1)} \leq \kappa
	\] 
	which is indeed satisfied under the conditions of the theorem. Therefore condition~\eqref{eq:bsuff} and the conclusions of Theorem~\ref{thm:main} hold true. Theorem~\ref{thm:groe} follows. 
\end{proof}

\begin{proof} [Proof of Theorem~\ref{thm:NF}]
	In Theorem~\ref{thm:main} we make the ansatz
	\[
		 m(x_1,\ldots, x_s) =  \kappa^{s-1} \prod_{1\leq i < j \leq s} \bigl(1+ f(x_i,x_j)\bigr). 
	\] 
	Then
	\[
		(R_\rho m)(x_1) = \frac 1\kappa \Psi(\rho \kappa) \quad (x_1\in \R^d).
	\]
	The numerator on the  left-hand side of~\eqref{eq:suff1} is equal to 
	\begin{multline}\label{eq:num} 
		\prod_{i=1}^s \bigl(1+ f(x_i,x_j)\bigr)
		\Biggl(\kappa^{s-2} 
		+ \sum_{j=1}^\infty \frac{1}{j!}\int_{(\R^d)^j}\prod_{i=1}^j |f(x_1,y_i)| \\
		\times \prod_{\substack{1\leq i \leq s,\\ 1\leq j \leq n}}(1+ f(x_i,y_j)) \prod_{1\leq i < j \leq s}(1+ f(y_i,y_j)) \kappa ^{s-2+j} \rho^j\dd \vect y\Biggr).
	\end{multline} 
	For non-negative pair potentials, we have $1+f = \e^{-v}\leq 1$. We bound the terms $1+ \Psi(x_i,y_j)$  on the second line of the previous displayed equation by $1$ and see that~\eqref{eq:num} is bounded from above by 
	\[
		\kappa^{s-2} \prod_{1\leq i < j \leq s} \bigl(1+f(x_i,x_j)\bigr) \Psi(\rho \kappa). 
	\] 
	Therefore Theorem~\ref{thm:main} is applicable if 
	\[
		\frac1\kappa (\Psi(\kappa\rho)-1)<1,\quad \frac{\Psi(\kappa\rho)}{1- \frac1\kappa (\Psi(\kappa\rho)-1)}\leq \kappa
	\] 
	which is indeed true if $2\Psi(\kappa \rho) - 1\leq \kappa$. 
\end{proof}

\section{Relation with connected graphs and activity expansions}\label{sec:activity}

Remember the set $\mathcal C_n$ of connected graphs with vertex set $[n]$. The \emph{Ursell function} is
\[
	\varphi_n^\mathsf T(x_1,\ldots, x_n) = \sum_{G\in \mathcal C_n} w(G;x_1,\ldots, x_n).
\] 
Given a measure $z$ on $(\mathbb X,\mathcal X)$, we define a new measure $\rho_1^z$ on $(\mathbb X,\mathcal X)$ by 
\be \label{eq:rhoconn}
	\rho_1^z(\dd q) = z(\dd q) \Bigl(1+ \sum_{n=1}^\infty \frac{1}{n!}\int_{\mathbb X^n} \varphi_{n+1}^\mathsf T(q,x_1,\ldots,x_n) z^n(\dd \vect x)\Bigr)
\ee
whenever the right-hand side is absolutely convergent. We may also view $\rho_1^z$ as a measure-valued formal power series; it is the exponential generating function for rooted connected graphs. When  convergent the measure $\rho_1^z$ corresponds to the intensity measure (one-particle density) of a grand-canonical Gibbs measure at activity $z$, see Appendix~\ref{app:gibbs}.  The inverse of the map $z\mapsto \rho_1^z$ is given by 
\be \label{eq:zdef}
	z(\dd q) = \rho(\dd q) \exp\Biggl( - \sum_{n=1}^\infty \frac{1}{n!}\int_{\mathbb X^n} D_{n+1}(q,x_1,\ldots, x_n)\rho^n(\dd \vect x)\Biggr),
\ee
see \cite[Theorem 3.4]{jansen-kuna-tsagkaro2019} for a rigorous statement for inhomogeneous systems, addressing issues of convergence as well. Theorem~\ref{thm:inversion} below is similar to the latter but works under a different conditions on $\rho$. 
Precisely, we assume that for some measurable function $m:\sqcup_{s\in \N}\mathbb X^s\to\R_+$, the measure $\rho$ satisfies the conditions from Theorem~\ref{thm:main} (i.e., $(R_\rho m)(x_1) <1$ on $\mathbb X$ and ~\eqref{eq:suff1}) and in addition 
\be \label{eq:suppsuff}
	\frac{1+ (R_\rho m)(x_1)}{1- (R_\rho m)(x_1)} \leq m(x_1) 
\ee
for all $x_1\in \mathbb X$. Notice that Theorem~\ref{thm:main} guarantees the absolute convergence of the series in~\eqref{eq:zdef}, with
\[
	z(\dd q) \leq \rho(\dd q) \e^{ \bar d (q;\rho)} \leq \rho(\dd q)\bigl( 1- R_\rho m(q)\bigr)^{-1}.
\] 
The additional condition~\eqref{eq:suppsuff} ensures that not only are all the density expansions convergent, but in addition the activity is in the domain of convergence of the activity expansions, see Lemma~\ref{lem:satisfied}.

\begin{remark} Condition~\eqref{eq:suppsuff} is automatically satisfied for weight functions $m(x_1,\ldots, x_s) = \kappa^s$ with $\kappa$ as in Theorem~\ref{thm:LP}, i.e., in the Lebowitz-Penrose domain. This extends to the inhomogeneous weights $m(x_1,\ldots, x_s) =\exp( b(x_1)+\cdots + b(x_s))$ with $b$ and $\rho$ as in~\eqref{eq:suff2}. However condition~\eqref{eq:suppsuff} excludes weight functions $m$ with $m(x_1) =1$ for all $x_1\in \mathbb X$, as used in the proof of Theorems~\ref{thm:groe} and~\ref{thm:NF} (Groeneveld and Nguyen-Fern{\'a}ndez domain).
\end{remark}

\begin{theorem} \label{thm:inversion}
	Let $\rho$ and $m$ satisfy the conditions from Theorem~\ref{thm:main} and in addition condition~\eqref{eq:suppsuff}. Define the measure $z$ by~\eqref{eq:zdef}. 
	Then the series defining $\rho_1^z(\dd q)$ is absolutely convergent and 
	$\rho_1^z = \rho$. 
\end{theorem} 

Before we turn to the proof, we address factorial moment measures (a.k.a.\ $s$-point correlation functions or distribution functions).  Let $W$ and $B$ be two finite disjoint sets, with $W$ non-empty. Elements of $W$ are called white, elements of $B$ are called black. We define $\mathcal C(W,B)$ as the set of graphs with vertex set $W\cup B$ such that every black  vertex is connected to at least one white vertex by a  path in $G$. 
Set
\[
	\varphi\bigl( (x_i)_{i\in W}; (x_j)_{j\in B}) := \sum_{G\in \mathcal C(W,B)} w\bigl(G; (x_i)_{i\in W\cup B}\bigr).
\] 
We define a function $\alpha(\cdot;z):\sqcup_{s\in \N}\mathbb X\to \R_+$ by 
\begin{multline}\label{eq:alphadef}
 \alpha (x_1,\ldots, x_s;z)  
	=  \varphi(x_1,\ldots, x_s;\varnothing)\\
	+ \sum_{n=1}^\infty \frac{1}{n!}\int_{\mathbb X^n} \varphi(x_1,\ldots, x_s;x_{s+1},\ldots, x_{s+n})
	z(\dd x_{s+1})\cdots z(\dd x_{s+n}).
\end{multline}
The convergence of $\alpha$ is addressed in Theorem~\ref{thm:correlations} below. Assuming absolute convergence of the series, we define a family of measures $\rho_s^z$ on $(\mathbb X^s,\mathcal X^s)$, $s\in \N$,  by 
\[
	\rho_s^z\bigl( \dd (x_1,\ldots, x_s)\bigr) = \alpha(x_1,\ldots, x_s;z) z(\dd x_1)\cdots z(\dd x_s). 
\]
The measures $\rho_s^z$ are the factorial moment measures ($s$-point correlation functions) of a Gibbs measure at activity $z$, see \cite[Section 4]{stell1964} and  Proposition~\ref{prop:corr-z}. Let 
\begin{multline} \label{eq:hdef}
 g(x_1,\ldots, x_s;\rho)  
	=  \psi(x_1,\ldots, x_s;\varnothing)\\
	+ \sum_{n=1}^\infty \frac{1}{n!}\int_{\mathbb X^n} \psi(x_1,\ldots, x_s;x_{s+1},\ldots, x_{s+n})
	\rho(\dd x_{s+1})\cdots \rho(\dd x_{s+n}).
\end{multline}

\begin{theorem}\label{thm:correlations}
	Let $\rho$ and $z$ be as in Theorem~\ref{thm:inversion}, and $\rho_1^z$ as in~\eqref{eq:rhoconn}.  Then the series defining $\alpha(x_1,\ldots,x_s;z)$ are absolutely convergent, moreover
	\[
		\rho_s^z\bigl( \dd(x_1,\ldots, x_s)\bigr) = g(x_1,\ldots,x_s;\rho_1^z) \rho_1^z(\dd x_1)\cdots \rho_1^z(\dd x_s) \quad (s\geq 2).
	\] 
\end{theorem} 
\noindent (The absolute convergence of $g(x_1,\ldots,x_s;\rho_1^z)$ follows from the equality $\rho_1^z=\rho$ proven in Theorem~\ref{thm:inversion}, the conditions on $\rho$, and Theorem~\ref{thm:main}.)

The proofs of Theorems~\ref{thm:inversion} and~\ref{thm:correlations} build on a  similar convergence result for generating functions $\alpha$ from~\eqref{eq:alphadef} from~\cite{jansen-kolesnikov2020} which we briefly recall. For $\tilde m:\sqcup_{s\in \N}\to \R_+$ a measurable function, we define $\mathbf T_z \tilde m:\sqcup_{s\in \N}\mathbb X^s\to \R_+\cup\{\infty\}$ by 
\[
	(\mathbf T_z \tilde m)(x_1) = 1+ \sum_{j=1}^\infty \frac{1}{j!}\int_{\mathbb X^j} \prod_{i=1}^j|f(x_1,y_i)| \tilde m(y_1,\ldots, y_j) z^n(\dd \vect y)
\]
and for $s\geq 2$
\begin{multline}
	(\mathbf T_z \tilde m)(x_1,\ldots, x_s) = \prod_{i=2}^s \bigl(1+ f(x_1,x_i)\bigr) \Bigl( 
	\tilde m(x_2,\ldots, x_s)\\
		+ \sum_{j=1}^\infty \frac{1}{j!}\int_{\mathbb X^j} \prod_{i=1}^j|f(x_1,y_i)|\, \tilde m(x_2,\ldots, x_s, y_1,\ldots, y_j) z^j(\dd \vect y)\Bigr). 
\end{multline} 
We define $\bar \alpha:\sqcup_{s\in \N}\mathbb X^s\to \R_+$ by a formula similar to~\eqref{eq:alphadef} but with absolute values around $\varphi$. 
 Theorem~\ref{thm:main} was modeled after the following theorem. 

\begin{theorem} \label{thm:activity} \cite{jansen-kolesnikov2020}
	Let $z$ be a measure on $(\mathbb X,\mathcal X)$. Suppose there exists a non-negative measurable function $\tilde m:\sqcup_{s\in \N}\mathbb X^s\to \R_+$ such that 
	\be \label{eq:actisuff}
		\mathbf T_z \tilde m\leq \tilde m
	\ee
	on $\sqcup_{s\in \N}\mathbb X^s$.
	Then $\bar \alpha\leq \tilde m<\infty$ on $\sqcup_{s\in \N}\mathbb X^s$.
\end{theorem} 

\noindent The result is formulated in~\cite{jansen-kolesnikov2020} for non-negative potentials only, in fact as a necessary and sufficient condition, the sufficient condition is easily extended to general pair potentials. 

The first step in the proofs of Theorems~\ref{thm:inversion} and~\ref{thm:correlations} is to check that the conditions of Theorem~\ref{thm:activity} are met, thus guaranteeing the absolute convergence of the series defining $\rho_s^z$, $s\geq 1$. 
Set
\be \label{eq:tildem}
	\tilde m(x_1,\ldots, x_s):= m(x_1,\ldots, x_s) \prod_{i=1}^s \bigl( 1 - R_\rho m(x_i)\bigr).
\ee

\begin{lemma} \label{lem:satisfied}
	Under the conditions of Theorem~\ref{thm:inversion}, the activity $z$ defined in~\eqref{eq:zdef} and the function $\tilde m$ from~\eqref{eq:tildem} satisfy condition~\eqref{eq:actisuff} from Theorem~\ref{thm:activity}.
\end{lemma}

\begin{proof} 
	Set $\bar z(\dd x) = \rho(\dd x_1) (1- R_\rho m(x_1))^{-1}$ and notice that 
	\[
		\tilde m(x_1,\ldots, x_s) \bar z^s(\dd \vect x) =	m(x_1,\ldots, x_s) \rho^s(\dd \vect x).
	\]
	Using~\eqref{eq:suppsuff}, we get 
	\[
		\mathbf T_z \tilde m(x_1) \leq \mathbf T_{\bar z} \tilde m(x_1) =1+ R_\rho m(x_1) \leq (1- R_\rho m(x_1)) m(x_1) = \tilde m(x_1).
	\] 
	For $s\geq 2$, we have the inequality 
	\begin{align*}
		\mathbf T_z \tilde m(x_1,\ldots, x_s) &\leq \mathbf T_{\bar z} \tilde m(x_1,\ldots, x_s)\\
			& = \prod_{i=1}^s \bigl(1- R_\rho m(x_i)\bigr) \mathbf K_\rho m (x_1,\ldots, x_s)\\
			&\leq \prod_{i=1}^s \bigl(1- R_\rho m(x_i)\bigr) m (x_1,\ldots,x_s)\\
			&= \tilde m(x_1,\ldots,x_s). \qedhere
	\end{align*} 
\end{proof} 

\begin{proof}[Proof of Theorem~\ref{thm:inversion}]
	We follow the proof of Theorem~3.4 in \cite{jansen-kuna-tsagkaro2019}. 
	The absolute convergence of the series defining $\rho_1^z$ is guaranteed by Theorem~\ref{thm:activity} and Lemma~\ref{lem:satisfied}, it remains to show that $\rho_1^z = \rho$. 
	First we note that 
	\[
		\rho_1^z(\dd q) = z(\dd q) \exp\Biggl( \sum_{n=1}^\infty \frac{1}{n!}\int_{\mathbb X^n} D_{n+1}(q,x_1,\ldots, x_n) \rho_1^z(\dd x_1)\cdots \rho_1^z(\dd x_n)\Biggr)
	\] 
	as an equality of formal power series in $z$ (see Lemma~3.8 in~\cite{jansen-kuna-tsagkaro2019} and the references therein, and \cite[Theorem 1.3]{leroux2004}), which implies  
	\be \label{eq:zr}
		z(\dd q) = \rho_1^z(\dd q) \exp\Biggl( -\sum_{n=1}^\infty \frac{1}{n!}\int_{\mathbb X^n} D_{n+1}(q,x_1,\ldots, x_n) \rho_1^z(\dd x_1)\cdots \rho_1^z(\dd x_n)\Biggr)
	\ee
	as an equality of formal power series. Eq.~\eqref{eq:zr} plays the role of the fixed point equation ($\mathsf{FP'}$) in~\cite[Section 2.3]{jansen-kuna-tsagkaro2019}. Let us write $\tilde \zeta[\rho](\dd q)$ for the measure-valued power series on the right-hand side of~\eqref{eq:zdef}, so that~\eqref{eq:zdef} becomes $z = \tilde \zeta[\rho]$. Let us also write $\rho_1[z](\dd q)$ instead of $\rho_1^z$. Then~\eqref{eq:zr} reads 
	\[
		z = \tilde \zeta \bigl[ \rho_1[z]\bigr]
	\]
	or equivalently, $\mathrm{id} = \tilde \zeta\circ \rho_1$ with the notion of composition of measure-valued formal power series given in \cite[Appendix A]{jansen-kuna-tsagkaro2019}. Proceeding as in \cite[Lemma 2.1]{jansen-kuna-tsagkaro2019} one finds a measure-valued formal power series $\zeta[\rho]$ such that $\rho_1 \circ\zeta = \mathrm{id}$. It follows that 
	\[
	   \zeta = \mathrm{id}\circ \zeta = (\tilde \zeta \circ \rho_1) \circ \zeta = \tilde \zeta \circ (\rho_1\circ \zeta) = \tilde \zeta \circ \mathrm{id} = \tilde \zeta,
	\]
	hence also $\rho_1\circ \tilde \zeta [\rho] = \rho$ as an equality of formal power series in $\rho$, i.e., 
	\begin{align}
		\rho(\dd q) &= \tilde \zeta[\rho](\dd q) \Biggl( 1+ 
		\sum_{n=1}^\infty \frac{1}{n!}\int_{\mathbb X^n} \varphi_{n+1}^\mathsf T(q,x_1,\ldots, x_n) 
		\prod_{i=1}^n\tilde \zeta[\rho](\dd x_i)\Biggr)\notag \\
		  & = \rho(\dd q) \e^{- d(q;\rho)}\Bigl( 1+ \sum_{n=1}^\infty \frac{1}{n!}\int_{\mathbb X^n} \varphi_n^\mathsf T(q,x_1,\ldots, x_n) \prod_{i=1}^n \e^{ - d(x_i,\rho)} \rho^n(\dd \vect x)\Bigr) \label{eq:almostthere}
	\end{align} 
	 as an equality of formal power series in $\rho$, where 
	\[
		d(q,\rho):= \sum_{n=1}^\infty \frac{1}{n!}\int_{\mathbb X^n} D_{n+1}(q,x_1,\ldots, x_n)\rho^n(\dd \vect x).
	\] 
The right-hand side of~\eqref{eq:almostthere} is a power series that is absolutely convergent under our assumptions on $\rho$. 
	Indeed if we put absolute values around the expansion coefficients of $\exp( - d(x_i,\rho))$, we obtain simply $\exp( \bar d(x_i,\rho))$ which is smaller than $(1- R_\rho m(x))^{-1}$. The proof of Lemma~\ref{lem:satisfied} shows that the measure $\bar z(\dd x) = (1- R_\rho m(x_i))^{-1} \rho(\dd x)$ is in the domain of convergence of the series with coefficients $|\varphi_{n+1}^\mathsf T(q,x_1,\ldots, x_n)|$. Thus the right-hand side of~\eqref{eq:almostthere} is absolutely convergent and  the equality holds true as an equality of convergent expressions.
\end{proof} 

\begin{proof}[Proof of Theorem~\ref{thm:correlations}]
	The absolute convergence of $\alpha(x_1,\ldots, x_s;z)$ follows from Theorem~\ref{thm:activity} and Lemma~\ref{lem:satisfied}. 	Standard combinatorial considerations in the spirit of ~\cite{stell1964}, \cite{leroux2004}, and \cite[Lemma 3.8]{jansen-kuna-tsagkaro2019} show that 
	\be \label{eq:first}
		\alpha(x_1,\ldots,x_s;z) = g(x_1,\ldots, x_s;\rho_1^z)\prod_{i=1}^s \e^{d(x_i;\rho_1^z)}
	\ee
	as an equality of formal power series in $z$. This reflects that every graph $G\in \mathcal C(W,B)$ with $W$, $B$ two finite non-empty sets, splits into a graph $G'\in \mathcal D(W,B')$, with $B'\subset B$, and a collection of connected graphs $(G''_i)_{i\in W\cup B'}$ such that each $G''_i$ contains the vertex $i$. Writing $V(G''_i)$ for the vertex set of $G''_i$, we have $\bigcup_{i\in W\cup B'} ( V(G''_i)\setminus\{i\}) = B\setminus B'$. The set $V(G''_i)\setminus \{i\}$ may be empty, it contains precisely those black vertices for which every path in $G$ to a white vertex different from $i$ passes through $i$. 
	
	In Eq.~\eqref{eq:first} we insert $z= \tilde \zeta[\rho]$ and obtain the formal power series equality
	\be \label{eq:second}
		\alpha(x_1,\ldots,x_s;\tilde \zeta[\rho])\prod_{i=1}^s \e^{-d(x_i;\rho)} = g(x_1,\ldots, x_s;\rho).
	\ee
	Our assumptions guarantee that the latter also holds true as an equality of convergent expressions. It follows that 
	\[
		\alpha(x_1,\ldots, x_s;z) z(\dd x_1)\cdots z(\dd x_s)= g(x_1,\ldots, x_s;\rho) \rho(\dd x_1)\cdots \rho(\dd x_s).
	\] 
	By Theorem~\ref{thm:inversion}, $\rho_1^z = \rho$ and the proof is complete. 
\end{proof} 

\appendix 

\section{Correlation functions of Gibbs measures}  \label{app:gibbs}

The classes of graphs $\mathcal C(W,B)$ and $\mathcal D(W,B)$ and their generating functions $\alpha(x_1,\ldots, x_s;z)$ and $g(x_1,\ldots, x_s;\rho)$ enter the activity and density expansions of correlation functions of grand-canonical Gibbs measures. This is well-known; for a derivation on the level  of formal power series, see e.g.~\cite{stell1964}. The convergence of density expansions of the correlation functions (in homogeneous systems) is addressed e.g.\ in \cite{lebowitz-penrose1964,bogoljubov-petrina-hacet1969}. In fact the correlation functions for homogeneous systems can be treated directly in the finite-volume canonical ensemble, see~\cite{pulvirenti-tsagkaro2015} for the two-point function and \cite{kuna-tsagkaro2018} for the truncated and direct correlation function and the Ornstein-Zernike equation.  

This appendix provides a statement adapted to our inhomogeneous setup (Proposition~\ref{prop:corr-rho}). In addition, we comment on the Kirkwood-Salsburg equation (Proposition~\ref{prop:last}). 
 Let $(\mathbb X,\mathcal X)$ be a Polish space, $\mathcal X_\mathrm b$ the collection of bounded measurable sets,  and $\Gamma$ the space of locally finite point configurations, i.e., 
\[
	\Gamma:= \bigl\{\eta\subset \mathbb X\mid \forall A\in \mathcal X_\mathrm b:\ \#(\eta \cap A) <\infty\bigr\}. 
\] 
The space $\Gamma$ is equipped with the $\sigma$-algebra $\mathscr F$ generated by the mappings $\gamma\mapsto \#(\eta\cap A)$, $A\in \mathcal X_\mathrm b$. 
The interaction of a point at $x\in \mathbb X$ with a configuration $\eta\in \Gamma$ is
\[
	H(x\mid \eta) := \begin{cases}
			\sum_{y\in \eta} v(x,y), &\quad \text{if the sum is absolutely convergent},\\
			\infty, &\quad \text{else}. 
		\end{cases} 
\] 
More generally, we set
\[
	H(x_1,\ldots, x_s\mid \eta) := \sum_{1\leq i <j \leq s} v(x_i,x_j) + \sum_{i=1}^s H(x_i\mid \eta).
\] 
In the following $z$ is a locally finite, diffuse measure on $(\mathbb X,\mathcal X)$, i.e., $z(\{x\}) = 0$ for all $x\in \mathbb X$ and $z(A)<\infty$ for all $A\in \mathcal X_\mathrm b$.

\begin{definition} 
A probability measure $\mathsf P$ on $(\Gamma,\mathscr F)$ is a Gibbs measure at activity $z$ for the pair potential $v$ if
\be \label{eq:gnz} \tag{$\mathsf{GNZ}$}
	\mathsf E\Bigl[ \sum_{x\in \eta} F(x,\eta)\Bigr] =\int z(\mathrm d x) \mathsf E\bigl[ \e^{- H(x\mid\eta)}F\bigl(x, \eta\setminus \{x\}\bigr)\bigr].
\ee
for all measurable functions $F:\mathbb X\times \Gamma\to \R_+$, with $\mathsf E[f]= \int_\Gamma f\dd \mathsf P$ the expected value with respect to $\mathsf P$. 
\end{definition} 

This definition uses the GNZ equation (named after Georgii, Nguyen, and Zessin) instead of the DLR equation (Dobrushin, Lanford, Ruelle) more familiar in mathematical physics; for the equivalence of the two characterizations, the reader is referred to the recent survey \cite{dereudre2019survey}. 

The intensity measure (one-particle density) $\rho_1$ of a Gibbs point process is absolutely continuous with respect to the activity, with
\[
	\rho_1(\dd x) = z(\dd x) \mathsf E[\e^{- H(x\mid \eta)}].
\] 
More generally,  the factorial moment measure of order $s$ (also known as $s$-point correlation function or distribution function) satisfies
\be \label{eq:rhosdens}
	\rho_s\bigl( \mathrm d(x_1,\ldots,x_s)\bigr) = z(\mathrm d x_1)\cdots z(\mathrm d x_s) \, \mathsf E\Bigl[ \e^{- H(x_1,\ldots,x_s\mid \eta)} \Bigr]. 
\ee
For simplicity we stick to the case where the grand-canonical partition function 
\[
	\Xi(z) := 1+ \sum_{n=1}^\infty \frac{1}{n!}\int_{\mathbb X^n} \e^{- \sum_{1\leq i < j \leq n} v(x_i,x_j)} z^n(\dd \vect x) 
\]
is finite, which amounts to working in finite volume. Under this condition the (grand-canonical) Gibbs measure is uniquely defined and given by 
\begin{multline} \label{eq:finitevol}
	\mathsf  P(\eta\in A) \\
	= \frac{1}{\Xi(z)} \Bigl( \1_A(\varnothing) + \sum_{n=1}^\infty \frac{1}{n!} \int_{\mathbb X^n} \1_A(\{x_1,\ldots, x_n\}) \e^{- \sum_{1\leq i < j \leq n} v(x_i,x_j)} z^n(\dd \vect x)\Bigr). 
\end{multline} 

\begin{prop}\label{prop:corr-z}
		Let $z$ be a locally finite, diffuse measure on $(\mathbb X,\mathcal X)$ with $\Xi(z)<\infty$. Suppose that $z$ satisfies the sufficient convergence condition from Theorem~\ref{thm:activity}. Then 
 the factorial moment measures of the uniquely defined Gibbs measure $\mathsf P$ are given by 
		\be \label{eq:factmom}
			\rho_s\bigl( \dd (x_1,\ldots, x_s)\bigr) = \alpha(x_1,\ldots,x_s;z) z(\dd x_1)\cdots z(\dd x_s)
		\ee
		with $\alpha(x_1,\ldots, x_s;z)$ defined in~\eqref{eq:alphadef}. 
\end{prop}

\begin{remark}
	The proposition uses only the absolute convergence of the series for $\alpha(\vect x;z)$, the precise sufficient convergence condition does not matter. 
\end{remark}


\begin{prop} \label{prop:corr-rho}
	Let $\rho$ be a locally finite, diffuse measure on $(\mathbb X,\mathcal X)$. Suppose 	
	that $\rho$ satisfies the conditions from Theorem~\ref{thm:inversion} and that the measure $z$ defined by 
	\be \label{eq:zdef2} 
		 z(\dd q) = \rho(\dd q) \exp\Bigl( - \sum_{n=1}^\infty \frac{1}{n!}\int_{\mathbb X^n} D_{n+1}(q,x_1,\ldots, x_n) \rho^n(\dd \vect x)\Bigr)
	\ee
	satisfies $\Xi(z)<\infty$. Then the intensity measure of the uniquely defined Gibbs measure $\mathsf P $ is $\rho_1 = \rho$ and the correlation functions are 
	\[
		\rho_s\bigl(\dd (x_1,\ldots, x_s)\bigr)= g(x_1,\ldots,x_s;\rho) \rho(\dd x_1)\cdots \rho(\dd x_s)
	\]
	with $g(x_1,\ldots, x_s;\rho)$ defined in~\eqref{eq:hdef}. 
\end{prop} 

\begin{proof}[Proof of Proposition~\ref{prop:corr-rho}]
	The proposition is an immediate consequence of Proposition~\ref{prop:corr-z} and Theorems~\ref{thm:inversion} and~\ref{thm:correlations}. 		
\end{proof} 

\begin{proof} [Proof of Proposition~\ref{prop:corr-z}]
  The statement is well-known on the level of formal power series~\cite[Section 4]{stell1964}, we present a self-contained proof for the reader's convenience. 
	In view of Eq.~\eqref{eq:rhosdens} we have to evaluate the expected value of $\exp( - H(\vect x\mid \eta))$, which we do with the help of~\eqref{eq:finitevol}. Let 
	\[
		H(x_1,\ldots, x_s) = \sum_{1\leq i < j \leq s} v(x_i-x_j).
	\]
	We want to evaluate 
	\begin{multline} \label{eq:multipinned}
		\mathsf E\bigl[\e^{- H(x_1,\ldots, x_s\mid \eta)}\bigr] \\
		 = \frac{1}{\Xi(z)}\Bigl( \e^{- H(x_1,\ldots, x_s)} 
		 	+ \sum_{n=1}^\infty \frac{1}{n!}\int_{\mathbb X^n} \e^{- H(x_1,\ldots, x_{s+n})} z(\dd x_{s+1})\cdots z(\dd x_{s+n})\Bigr).  
	\end{multline}
	The integrand is equal to 
	\[
		\sum_{G\in \mathcal G_{s+n}}\prod_{\{i,j\}\in E(G)} f(x_i,x_j). 
	\]  
	Set $W:= \{1,\ldots, s\}$ and $B:= \{s+1,\ldots,s+n\}$. Given $G\in \mathcal G_{s+n}$, let $B_1\subset B$ be the set of vertices that are connected to $W$ by a path in $G$, and $B_2:= B\setminus B_1$. Then there is no edge that links $B_2$ to $W\cup B_1$, and $G$ splits into two graphs $G_1$ and $G_2$ with respective vertex sets $W\cup B_1$ and $B_2$. These combinatorial considerations set up a one-to-one correspondence  between graphs $G\in \mathcal G_{s+n}$ and pairs $(G_1,G_2)$ such that 
	\begin{enumerate} [(i)]
		\item [(i)] $G_1$ is a graph in $\mathcal C(W,B_1)$ with $B_1\subset B$. 
		\item [(ii)] $G_2$ is a graph with vertex set $B_2=B\setminus B_1$.  
	\end{enumerate} 
	It follows that 
	\[ 
		\e^{- H(x_1,\ldots, x_{s+n})}= \sum_{(B_1,B_2)} \varphi(x_1,\ldots, x_s;(x_i)_{i\in B_1})\, \e^{- H((x_j)_{j\in B_2})}
	\]
	where the summation is over pairs $(B_1,B_2)$ of disjoint sets with $B_1\cup B_2= B=\{s+1,\ldots, s+n\}$. As a consequence the term in parentheses in~\eqref{eq:multipinned} is equal to 
	\begin{multline}
		\varphi(x_1,\ldots, x_s;\varnothing) + \sum_{n=1}^\infty \frac{1}{n!} \int_{\mathbb X^n} \sum_{I\subset [n]} \varphi(x_1,\ldots, x_s;\vect y_{I}) \e^{- H(\vect y_{[n]\setminus I})} z^n(\dd \vect y)\\
		= \alpha(x_1,\ldots, x_s;z) \Xi(z). 
	\end{multline}
	The absolute convergence of $\alpha(x_1,\ldots, x_s;z)$ is guaranteed by Theorem~\ref{thm:activity}. 
	The factor $\Xi(z)$ is cancelled by the factor $1/\Xi(z)$ in~\eqref{eq:multipinned} and we obtain 
	\[
		\mathsf E\bigl[\e^{- H(x_1,\ldots, x_s\mid \eta)} \bigr]
			= \alpha(x_1,\ldots, x_s;z). 
	\] 
	The proposition now follows from~\eqref{eq:rhosdens}. 
\end{proof} 

We conclude with two remarks related Kirkwood-Salsburg equation, valid for all Gibbs measures i.e. also in infinite volume ($\Xi(z)=\infty$). The first remark explains how the convergence condition in Theorem~\ref{thm:main} arises from the Kirkwood-Salsburg equation and can be understood without any reference to recurrence relations of weighted graphs. As is well-known, under some additional growth conditions on its correlation functions a probability measure is a Gibbs measure if and only if the correlation functions satisfy the Kirkwood-Salsburg equations \cite{ruelle1970superstable, kondratiev-kuna2003}. In our notation the latter read 
\be \label{eq:ksr}
	\rho_1(\dd x_1) = z(\dd x_1) \Bigl( 1+ \sum_{n=1}^\infty \frac{1}{n!}\int_{\mathbb X^n} \prod_{i=1}^n f(x_1,y_i) \rho_n\bigl( \dd(y_1,\ldots, y_n)\bigr)\Bigr)
\ee
and for $s\geq 2$, 
\begin{multline}\label{eq:ks} 
	\rho_s\bigl( \dd( x_1,\ldots, x_s)\bigr) 
		= \prod_{i=2}^s (1+f(x_1,x_i))z(\dd x_1)  \Biggl( \rho_{s-1}\bigl(\dd( x_2,\ldots, x_s)\bigr) \\
		+ \sum_{n=1}^\infty \frac{1}{n!}\int_{\mathbb X^n} \prod_{i=1}^n f(x_1,y_i) \rho_{s-1+n}\bigl( \dd(x_2,\ldots,x_s,y_1,\ldots, y_n)\bigr)\Biggr).
\end{multline} 
The last equation uses somewhat abusive notation, 
 some readers may prefer the equation
\begin{multline}\label{eq:ks-tested} 
	\int_{\mathbb X^s} F(x_1,\ldots,x_s) \rho_s\bigl( \dd( x_1,\ldots, x_s)\bigr) 
		=\prod_{i=2}^s (1+f(x_1,x_i))\Biggl( \int_{\mathbb X^s} F(\vect x) (z\otimes \rho_{s-1})(\dd \vect x)\\		
		+ \sum_{n=1}^\infty \frac{1}{n!}\int_{\mathbb X^{s+n}} F(x_1,\ldots, x_s) \prod_{i=1}^n f(x_1,x_{s+i}) (z\otimes \rho_{s-1+n})\bigl(\dd (x_1,\ldots, x_{s+n})\bigr)\Biggr)
\end{multline} 
valid for all measurable functions $F:\mathbb X^s\to \R_+$. Now, assuming that the term in parentheses on the right-hand side of Eq.~\eqref{eq:ksr} is non-zero for all $x_1\in\mathbb X$, we can invert the relation and express $z(\dd x_1)$ as a function of the correlation functions $\rho_s$. Plugging this relation into~\eqref{eq:ks} we obtain a set of integral equations for the correlation functions where the activity $z$ no longer appears. This translates into a set of integral equations for the Radon-Nikodym derivative of $\rho_s$ with respect to $\rho_1^{\otimes s}$. This set of integral equations is similar to the equation $\mathbf K_\rho m = m$ with $\mathbf K_\rho$ defined above Theorem~\ref{thm:main}, the only difference being that Mayer's $f$-function appears without absolute values and the geometric series for $(1-R_\rho m(x_1))^{-1}$ is replaced with $(1+ \sum_{n=1}^\infty \frac 1 {n!} \int_{\mathbb X^n} \prod_{i=1}^n f(x_1,y_i)\rho_n(\dd \vect y))^{-1}$.

Therefore the convergence condition $\mathbf K_\rho m \leq m$ is a natural counterpart of the Kirkwood-Salsburg equation, manipulated in such a way that the activity $z$ disappears. For homogeneous systems, such a set of equations can actually be derived directly from the canonical ensemble, see~\cite{bogoljubov-petrina-hacet1969}. 

The second remark concerns the additional condition~\eqref{eq:suppsuff} on $\rho$ needed in our proof of Proposition~\ref{prop:corr-rho}. As pointed out before Theorem~\ref{thm:inversion}, the additional condition is needed in order to ensure that not only do the density expansions converge, but in addition the activity defined as a function of $\rho$ is in the domain of convergence of the activity expansions. In self-explanatory notation, for homogeneous systems the condition ensures 
\[
	\rho \leq \sup_{z \leq R_\mathsf{May}} \rho(z)
\] 	
with $R_\mathsf{May}$ the radius of convergence of the Mayer series (activity expansion of the pressure). 
Clearly we would like to get rid of this condition. After all, as mentioned in the introduction, for non-negative interactions we expect that $\rho> \sup_{z\leq R_\mathsf{May}} \rho(z)$ and the density expansions should retain their physical relevance \emph{beyond} the domain of convergence of the activity expansions. Therefore a relevant observation is that the Kirkwood-Salsburg equations hold true without the additional assumption~\eqref{eq:suppsuff}.

\begin{prop}\label{prop:last}
	Under the conditions of Theorem~\ref{thm:main}, the family of measures $\rho_s$ \emph{defined} by $\rho_s(\dd(x_1,\ldots, x_s)) = g(x_1,\ldots, x_s;\rho) \rho(\dd x_1) \cdots\rho(\dd x_s)$ satisfies the Kirkwood-Salsburg equations at the activity $z$ given by~\eqref{eq:zdef2}.
\end{prop} 

The proof is based on the recurrence relation from Propositions~\ref{prop:2conn} and~\ref{prop:recursion}, it is similar to the proof of Proposition~\ref{prop:fpi}, the details are left to the reader.

Proposition~\ref{prop:last} still leaves us with the task of proving that there exists a probability measure $\mathsf P$ on $(\Gamma, \mathscr F)$ with factorial moment measures $\rho_s, s\geq 1$ (the Kirkwood-Salsburg equation would then imply that it is a Gibbs measure). For that one has to prove so-called \emph{Lenard positivity} following~\cite{lenard1975} or try to adapt the strategy from~\cite{nehring-poghosyan-zessin2013}. Another option is to impose slightly more stringent conditions on $\rho$ so that $z$ and $\rho_s$ fall into a domain where the Kirkwood-Salsburg equation has a unique solution.  We leave this as an open problem for future work. 

\medskip 

\noindent \emph{Acknowledgments.} I thank Tobias Kuna and Dimitrios Tsa\-gkarogiannis for many helpful discussions and literature guidance, and Roberto Fern{\'a}ndez and Nguyen Tong Xuan for comments on the work \cite{nguyen-fernandez2020}. 

%

\providecommand{\bysame}{\leavevmode\hbox to3em{\hrulefill}\thinspace}
\providecommand{\MR}{\relax\ifhmode\unskip\space\fi MR }
\providecommand{\MRhref}[2]{%
  \href{http://www.ams.org/mathscinet-getitem?mr=#1}{#2}
}
\providecommand{\href}[2]{#2}

\end{document}